\documentclass{llncs}

\usepackage{amssymb}
\usepackage{graphicx}
\usepackage{lineno}
\usepackage{subfigure}

\usepackage{times} 

\newcommand{\remove}[1]{}

\newcommand{\gracsimdrawlong}{\textsc{Geometric Rac Simultaneous drawing}}

\newcommand{\gracsim}{\textsc{GRacSim}}
\newcommand{\gracsimdrawing}{\textsc{GRacSim drawing}}
\newcommand{\gracsimdraw}{\textsc{GRacSim} drawing}
\newcommand{\gracsimdraws}{\textsc{GRacSim} drawings}

\newcommand{\sefelong}{\textsc{Simultaneous Embedding with Fixed Edges}}
\newcommand{\sefe}{\textsc{SEFE}}
\newcommand{\onesefe}{\textsc{1-SEFE}}
\newcommand{\ksefe}{\textsc{k-SEFE}}
\newcommand{\kpairsefe}{\mbox{\textsc{k-pair-SEFE}}}

\newcommand{\sgelong}{\textsc{Simultaneous Geometric Embedding}}
\newcommand{\sge}{\textsc{SGE}}

\newcommand{\threeplong}{\mbox{\textsc{3-Partition}}}
\newcommand{\threep}{\textsc{3P}}

\newcommand{\np}{$\mathcal{NP}$}
\newcommand{\npc}{\mbox{$\mathcal{NP}$-complete}}
\newcommand{\npcness}{\mbox{$\mathcal{NP}$-completeness}}
\newcommand{\nph}{\mbox{$\mathcal{NP}$-hard}}
\newcommand{\nphness}{\mbox{$\mathcal{NP}$-hardness}}

\spnewtheorem*{sketch}{Sketch of Proof}{\itshape}{\rmfamily}

\usepackage{url}
\urldef{\mailsa}\path|luca.grilli@unipg.it|

\mainmatter

\title{On the \nphness{} of \gracsimdrawing{} and\\ \ksefe{} Problems}

\author{Luca Grilli}

\institute{
Department of Engineering, University of Perugia, Italy\\
\mailsa}

\begin{document}

\maketitle
\thispagestyle{plain} 

\begin{abstract}
We study the complexity of two problems in simultaneous graph drawing.
The first problem, \gracsimdrawing{}, asks for finding a simultaneous geometric embedding of two graphs such that only crossings at right angles are allowed.
The second problem, \ksefe{}, is a restricted version of the topological simultaneous embedding with fixed edges (\sefe{}) problem, for two planar graphs, in which every private edge may receive at most $k$ crossings, where $k$ is a prescribed positive integer.
We show that \gracsimdrawing{} is \nph{} and that \ksefe{} is \npc{}.
The \nphness{} of both problems is proved using two similar reductions from \threeplong{}.
\end{abstract}

\section{Introduction}
\label{sec:intro}
The problem of computing a simultaneous embedding of two or more graphs has been extensively explored by the graph drawing community.
Indeed, besides its inherent theoretical interest~\cite{ad-scp-14,adf-seepg-13,adfpr-tsetgibcg-12,aefg-smliogpg-16,agkn-tpgse-12,adn-asefepbep-15,bdkw-sdpgrac-16,bkr-seeorpc-13,bkr-hgdv-13,br-spqoacep-13,cfglms-dpespg-15,ek-sepgfb-05,ekln-sgdlavs-05,egjpss-sgge-07,elm-svrpstguls-16,f-egsfe-06,fhk-sefbc-15,gjpss-sgefe-06,ghkr-dsegfb-14,hjl-tspcg2c-13,js-igsefe-09,s-ttphtpv-13}, it has several applications in dynamic network visualization, especially when a visual analysis of an evolving network is needed.
Although many variants of this problem have been investigated so far, a general formulation for two graphs can be stated as follows: Let $G_1 = (V_1, E_1)$ and $G_2 = (V_2, E_2)$ be two planar graphs sharing a \emph{common} (or \emph{shared}) subgraph $G = (V, E)$, where $V = V_1 \cap V_2$ and $E = E_1 \cap E_2$. 
\emph{Compute a planar drawing $\Gamma_1$ of $G_1$ and a planar drawing $\Gamma_2$ of $G_2$ such that the restrictions to $G$ of these drawings are identical.}
By overlapping $\Gamma_1$ and $\Gamma_2$ in such a way that they perfectly coincide on $G$, it follows that edge crossings may only occur between a \emph{private} edge of $G_1$ and a \emph{private} edge of $G_2$, where a \emph{private} (or \emph{exclusive}) edge of $G_i$ is an edge of $E_i \setminus E$ ($i=1,2$).

Depending on the drawing model adopted for the edges, two main variants of the simultaneous embedding problem have been proposed: \emph{topological} and \emph{geometric}.
The topological variant, known as \sefelong{} (or \sefe{} for short), allows to draw the edges of $\Gamma_1$ and $\Gamma_2$ as arbitrary open Jordan curves, provided that every edge of $G$ is represented by the same curve in $\Gamma_1$ and $\Gamma_2$.
Instead, the geometric variant, known as \sgelong{} (or \sge{} for short), imposes that $\Gamma_1$ and $\Gamma_2$ are two straight-line drawings.
The \sge{} problem is therefore a restricted version of \sefe{}, and it turned out to be ``too much restrictive'', i.e. there are examples of pairs of structurally simple graphs, such as a path and a tree~\cite{agkn-tpgse-12}, that do not admit an \sge{}.
Also, testing whether two planar graphs admit a simultaneous geometric embedding is \nph{}~\cite{egjpss-sgge-07}.
Compared with \sge{}, pairs of graphs of much broader families always admit a \sefe{}, in particular there always exists a \sefe{} when the input graphs are a planar graph and a tree~\cite{f-egsfe-06}.
In contrast, it is a long-standing open problem to determine whether the existence of a \sefe{} can be tested in polynomial time or not, for two planar graphs; though, the testing problem is \npc{} when generalizing \sefe{} to three or more graphs~\cite{gjpss-sgefe-06}.
However, several polynomial time testing algorithms have been provided under different assumptions~\cite{adfjkpr-tppeg-10,adfpr-tsetgibcg-12,bkr-hgdv-13,br-spqoacep-13,hjl-tspcg2c-13,s-ttphtpv-13}, most of them involve the connectivity or the maximum degree of the input graphs or of their common subgraph.

In this paper we study the complexity of the \gracsimdrawlong{} problem~\cite{abks-gracsdg-13} (\gracsimdrawing{} for short): a restricted version of \sge{}, which asks for finding a simultaneous geometric embedding of two graphs, such that all edge crossings must occur at right angles.
We show that \gracsimdrawing{} is \nph{} by a reduction from \threeplong{}; see Section~\ref{sec:gracsimnph}.
Moreover, we introduce a new restricted version of the \sefe{} problem, called \ksefe{}, in which every private edge may receive at most $k$ crossings, where $k$ is a prescribed positive integer.
We then show that \ksefe{} is \npc{} for any fixed positive $k$, to prove the \nphness{} we use a similar reduction technique as that for \gracsimdrawing{}; see Section~\ref{sec:ksefenpc}.


\section{Preliminaries}
\label{sec:prel}
Let $G = (V, E)$ be a simple graph.
A \emph{drawing} $\Gamma$ of $G$ maps each vertex of $V$ to a distinct point in the plane and each edge of $E$ to a simple Jordan curve connecting its end-vertices.
Drawing $\Gamma$ is \emph{planar} if no two distinct edges intersect, except at common end-vertices.
$\Gamma$ is a \emph{straight-line planar drawing} if it is planar and all its edges are represented by straight-line segments.
$G$ is \emph{planar} if it admits a planar drawing.
A planar drawing $\Gamma$ of $G$ partitions the plane into topologically connected regions called \emph{faces}.
The unbounded face is called the \emph{external} (or \emph{outer}) face; the other faces are the \emph{internal} (or \emph{inner}) faces.
A face $f$ is described by the circular ordering of vertices and edges that are encountered when walking along its boundary in clockwise direction if $f$ is internal, and in counterclockwise direction if $f$ is external.
A \emph{planar embedding} of a planar graph $G$ is an equivalence class of planar drawings that define the same set of faces for $G$.
A \emph{plane graph} is a planar graph with an associated planar embedding and a prescribed outer face.
Let $H$ be a plane graph.
The \emph{weak dual} $H^*$ of $H$ is a graph whose vertices correspond to the internal faces of $H$, and there is an edge between two vertices if the corresponding internal faces in $H$ share one or more edges.
A \emph{fan} is a graph formed by a path $\pi$ plus a vertex $v$ and a set of edges connecting $v$ to every vertex of $\pi$; vertex $v$ is called the \emph{apex} of the fan.
A \emph{wheel} is a graph consisting of a cycle $C$ plus a vertex $c$ and a set of edges connecting $c$ to every vertex of $C$; vertex $c$ is the \emph{center} of the wheel.

\newpage
\section{NP-hardness of \gracsimdrawing{}}\label{sec:gracsimnph}
In this section, we study the complexity of the following problem.
\\

\noindent
\begin{minipage}[t]{0.15\textwidth}
\raggedleft
\textsf{Problem}:
\end{minipage}
\begin{minipage}[t]{0.5\textwidth}
\end{minipage}
\begin{minipage}[t]{0.80\textwidth}
\raggedright
\gracsimdrawlong{} (\gracsimdrawing{})
\end{minipage}

\noindent
\begin{minipage}[t]{0.15\textwidth}
\raggedleft
\textsf{Instance}:
\end{minipage}
\begin{minipage}[t]{0.5\textwidth}
\end{minipage}
\begin{minipage}[t]{0.80\textwidth}
\raggedright
Two planar graphs $G_1 = (V, E_1)$ and $G_2 = (V, E_2)$ sharing a common subgraph $G = (V, E) = (V, E_1 \cap E_2)$.
\end{minipage}

\noindent
\begin{minipage}[t]{0.15\textwidth}
\raggedleft
\textsf{Question}:
\end{minipage}
\begin{minipage}[t]{0.5\textwidth}
\end{minipage}
\begin{minipage}[t]{0.80\textwidth}
\raggedright
Are there two straight-line planar drawings $\Gamma_1$ and $\Gamma_2$, of $G_1$ and $G_2$, respectively, such that (\emph{i}) every vertex is mapped to the same point in both drawings, and (\emph{ii}) any two crossing edges $e_1$ and $e_2$, with $e_1 \in E_1 \setminus E$ and $e_2 \in E_2 \setminus E$, cross only at right angle?
\end{minipage}

\begin{theorem}\label{th:gracsimnph}
Deciding whether two graphs have a \gracsimdrawing{} is~\nph{}.
\end{theorem}

\begin{proof}
We prove the \nphness{} by a reduction from \threeplong{} (\threep{}).\\

\noindent
\begin{minipage}[t]{0.15\textwidth}
\raggedleft
\textsf{Problem}:
\end{minipage}
\begin{minipage}[t]{0.5\textwidth}
\end{minipage}
\begin{minipage}[t]{0.80\textwidth}
\raggedright
\threeplong{} (\threep{})
\end{minipage}

\noindent
\begin{minipage}[t]{0.15\textwidth}
\raggedleft
\textsf{Instance}:
\end{minipage}
\begin{minipage}[t]{0.5\textwidth}
\end{minipage}
\begin{minipage}[t]{0.80\textwidth}
\raggedright
A positive integer $B$, and a multiset $A = \{a_1, a_2, \ldots, a_{3m}\}$ of $3m$ natural numbers with $B/4 < a_i < B/2$ ($1 \leq i \leq 3m$).
\end{minipage}

\noindent
\begin{minipage}[t]{0.15\textwidth}
\raggedleft
\textsf{Question}:
\end{minipage}
\begin{minipage}[t]{0.5\textwidth}
\end{minipage}
\begin{minipage}[t]{0.80\textwidth}
\raggedright
Can $A$ be partitioned into $m$ disjoint subsets $A_1, A_2, \ldots, A_m$, such that each $A_j$ ($1 \leq j \leq m$) contains exactly $3$ elements of $A$, whose sum is $B$?
\end{minipage}
\\
\\
\noindent
We recall that \threep{} is a \emph{strongly} \nph{} problem~\cite{gj-cigtnpc-79}, i.e., it remains \nph{} even if $B$ is bounded by a polynomial in $m$. Also, a trivial necessary condition for the existence of a solution is that $\sum_{i=1}^{3m} a_i = mB$, therefore it is not restrictive to consider only instances satisfying this equality.

We first give an overview of this reduction, then we describe in detail the construction for transforming an instance of \threep{} into an instance $\langle G_1, G_2 \rangle$ of \gracsimdrawing{}, and finally we prove that an instance of \threep{} is a \emph{Yes}-instance if and only if the transformed instance $\langle G_1, G_2 \rangle$ admits a \gracsimdraw{}.
\medskip


\begin{figure}[!tb]
\centering
\subfigure[Subdivided pumpkin]{
  \label{fi:gracsim-pumpkin}  
  \includegraphics[width=0.56\textwidth]{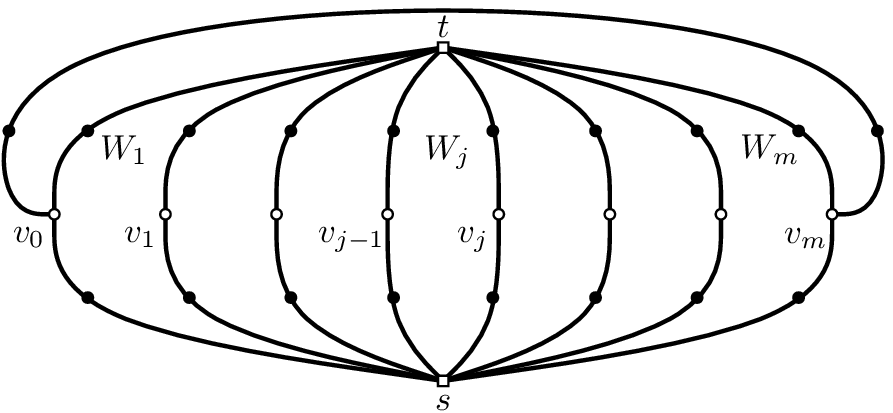}
}
\hspace{2mm}
\subfigure[Subdivided slice]{
  \label{fi:gracsim-slice}
  \includegraphics[width=0.24\textwidth]{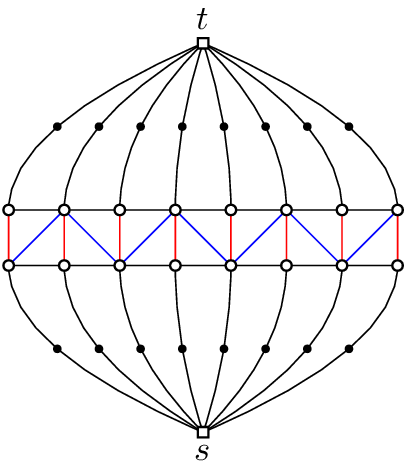}
}
\subfigure[Wedge $W_j$, transversal path $\pi_j$, and subdivided slices $S_{j1}$, $S_{j2}$, and $S_{j3}$]{
  \label{fi:gracsim-wedgezoom}
  \includegraphics[width=0.99\textwidth]{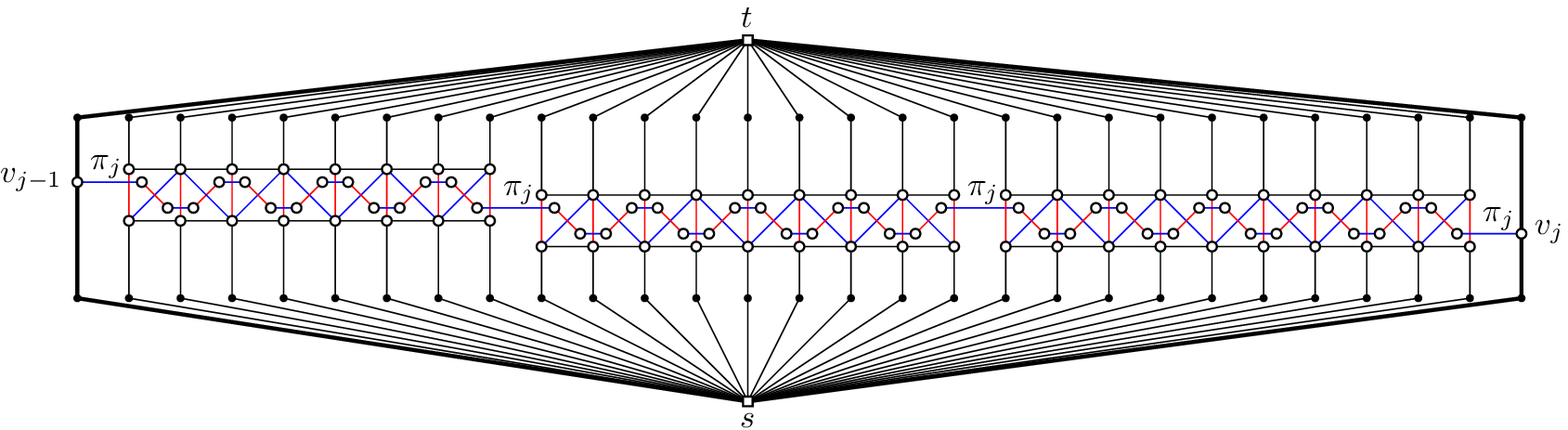}
}
\subfigure[\textsc{GRacSim} drawing]{
  \label{fi:gracsim-drawing}
  \includegraphics[width=0.99\textwidth]{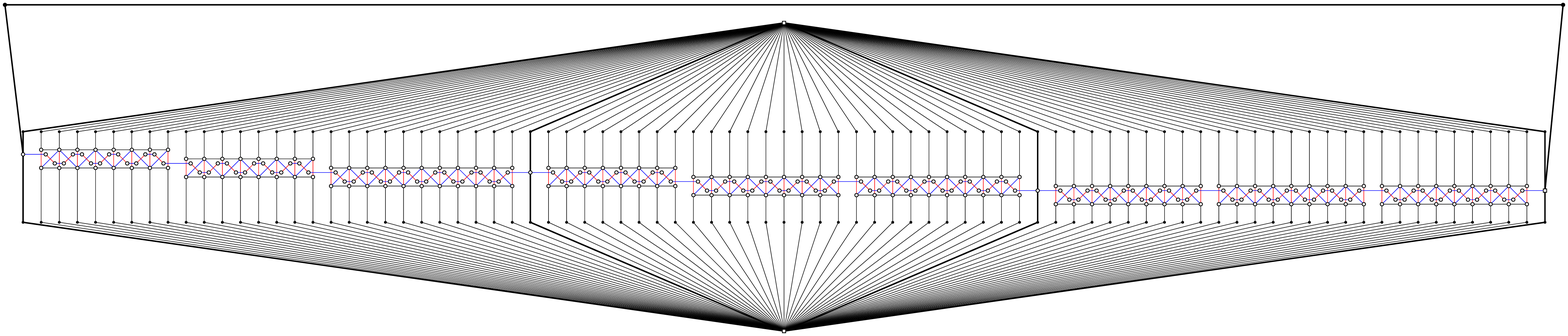}
}
\caption{
Illustration of (a) a subdivided pumpkin gadget and of (b) a subdivided slice gadget encoding integer $7$.
(c) A wedge $W_j$ of width~$24$, its transversal path $\pi_j$, and subdivided slices $S_{j1}$, $S_{j2}$, and $S_{j3}$ encoding integers~$7$,~$8$ and $9$, respectively.
Shared edges are colored black, those of the subdivided pumpkin with thick lines, while private edges of $G_1$ and of $G_2$ are colored blue and red, respectively.
(d) A (vertically stretched) \gracsimdraw{} of the transformed instance of \threep{}, when $m = 3$, $B = 24$ and $A = \{7,7,7,8,8,8,8,9,10\}$. Subdivided slices are drawn within wedges according to the following solution of \threep{}: $A_1 = \{7,7,10\}$, $A_2 = \{7,8,9\}$ and $A_3 = \{8,8,8\}$.
}
\label{fi:gracsim-pumpkin-tpath-slice-wedgezoom}
\end{figure}

\noindent
\textsc{Overview} The transformed instance $\langle G_1, G_2 \rangle$ of \gracsimdrawing{} is obtained by combining a \emph{subdivided pumpkin} gadget with $3m$ \emph{subdivided slice} gadgets and $m$ \emph{transversal paths}; see Fig.~\ref{fi:gracsim-pumpkin-tpath-slice-wedgezoom} for an illustration.
A \emph{pumpkin} gadget consists of a biclique $K_{2,m+1}$ plus an additional edge, called the \emph{handle}, that connects two vertices of the partite set of cardinality $m+1$; the two vertices of the other partite set are the \emph{poles of the pumpkin}.
A \emph{subdivided pumpkin} is a \emph{pumpkin} where each edge, other than the handle, is subdivided exactly once, while the handle is subdivided twice.
We remark that it is not strictly necessary to use a subdivided pumpkin instead of a normal pumpkin, the only reason is to exploit the subdivision vertices as bend points, in this way we get more readable and compact \gracsimdraws{}.
Hereafter, when it is not ambiguous, we will use the terms \emph{pumpkin} and \emph{slice} in place of \emph{subdivided pumpkin} and of \emph{subdivided slice}, respectively. 
All the edges of a pumpkin are \emph{shared} edges, that is, they belong to both graphs, therefore they cannot be crossed in any \gracsimdraw{}.
Moreover, any planar embedding of a subdivided pumpkin contains exactly two faces of degree seven and $m$ faces of degree eight, the latter are called $wedges$ and are the only faces incident to both poles.
Wedges are used to contain (\emph{subdivided}) \emph{slice} gadgets, which are $3m$ subgraphs attached to the two poles of the pumpkin, with no other vertices in common with each other and with the pumpkin.
Every slice has a ``width'' that suitably encodes a distinct element $a_i$ of $A$---recall that two distinct elements could be equal---and the structure of a slice is sufficiently ``rigid'' so that overlaps and nestings among slices cannot occur in a \gracsimdraw{}.

The basic idea of the reduction is to get the subsets $A_j$ ($1 \leq j \leq m$) of a solution of \threep{}, in case one exists, by looking at the slices in each wedge of a \gracsimdraw{}, which implies that every wedge must contain exactly three slices whose widths sum to $B$.
Of course, without introducing some further gadget, each wedge could contain even all slices, i.e. its width can be considered unlimited.
Hence, in order to make all wedges of the same width $B$, $m$ \emph{transversal paths} are attached to the pumpkin, one for each wedge.
Precisely, a \emph{transversal path} is an alternating path that connects the two vertices of a wedge other than the poles and the subdivision vertices, and it contains only non-shared edges that belong alternatively to $G_1$ and to $G_2$.
Therefore, the pumpkin plus the transversal paths form a subdivision of a maximal planar graph, which has a unique embedding (up to a choice of the external face).
Further, every transversal path has an ``effective length'' that encodes the integer $B$, which also establishes the width of the corresponding wedge.
Crossings between slices and transversal paths are thus unavoidable in a \gracsimdraw{}, because every transversal path splits its wedge into two parts, separating the two poles of the pumpkin; clearly, every slice crosses only one transversal path.
However, by choosing a suitable structure for the slices, it is possible to form only crossings that are allowed in a \gracsimdraw{}.
The key factor of the reduction is to make it possible if and only if each slice of width $a_i$ can cross a portion of its transversal path with an effective length greater than or equal to $a_i$.
In other words, the slice structure and the transversal path effective length are defined in such a way that, in a \gracsimdraw{}, (\emph{i}) every transversal path cannot cross more than three slices, and (\emph{ii}) the total width of the slices crossed by a same transversal path equals integer $B$, which yields a solution of \threep{}.

\noindent
\textsc{Construction} We now describe in detail a procedure to incrementally construct an instance $\langle G_1, G_2 \rangle$ of \gracsimdrawing{} starting from an instance of \threep{}.
At each step, this procedure adds one or more subgraphs (gadgets) to the current pair of graphs. 
As $G_1$ and $G_2$ have the same vertex set, for each added subgraph we will only specify which edges are shared and which are exclusive; the final vertex set will be known implicitly.

Start with a biclique $K_{2,m+1}$, and denote by $s,t$ and by $v_0,v_1, \ldots, v_m$ its vertices of the partite sets of cardinality~$2$ and $m+1$, respectively.
Add edge $h = (v_0,v_m)$ to the biclique, subdivide $h$ twice, and denote by $\pi_h$ the resulting $3$-edge path.
Then, for every $0 \leq j \leq m$, subdivide edge $(s, v_j)$ ($(t, v_j)$, respectively) exactly once, denote the subdivision vertex by $v_j^s$ ($v_j^t$, respectively) and the $2$-edge path obtained from this subdivision by $\pi_s(j)$ ($\pi_t(j)$, respectively).
The resulting graph $G_p$ is the \emph{subdivided pumpkin} and all its edges are shared edges, i.e. $G_p \subset G$; vertices $s$ and $t$ are the \emph{poles of the pumpkin}, while $\pi_h$ is called the \emph{subdivided handle} of the pumpkin.

Connect each pair of vertices $v_{j-1},v_j$ ($1 \leq j \leq m$) of $G_p$ with a \emph{transversal path} $\pi_j$, consisting of $2B + 1$ non-shared edges, so that edges in odd positions (starting from $v_{j-1}$) are private edges of $G_1$, while those in even positions are private edges of $G_2$; hence, every transversal path starts and ends with an edge of $G_1$ and has exactly $2B$ inner vertices.
Integer $B$ represents the \emph{effective length} of a transversal path, which is defined as half the number of its inner vertices.

For each integer $a_i \in A$, ($1 \leq i \leq 3m$) construct a (\emph{subdivided}) \emph{slice} $S_i$ by suitably attaching two fan subgraphs and by subdividing a subset of their edges as follows (see,~e.g.,~Fig.\ref{fi:gracsim-slice}).
Add a fan of $a_i + 2$ vertices with apex at pole $t$ and subdivide every edge incident to $t$ exactly once; denote the resulting subdivided fan by $F^t_i$.
Specularly, add a subdivided fan $F^s_i$ with apex at the other pole $s$, having the same number of vertices as $F^t_i$.
All the edges of $F^s_i$ and $F^t_i$ are shared edges, i.e. $F^s_i \cup F^t_i \subset G$.
Now, let $\pi^t_i$ and $\pi^s_i$ be the two paths of these fans, i.e. $\pi^t_i = F^t_i \setminus \{t\}$ and $\pi^s_i = F^s_i \setminus \{s\}$.
Visit path $\pi^t_i$ starting from one of its end-vertices and denote the $k$-th encountered vertex by $\pi^t_i(k)$ ($1 \leq k \leq a_i + 1$); in an analogous way define the $k$-th vertex $\pi^s_i(k)$ of path $\pi^s_i$.
For each $1 \leq k \leq a_i + 1$, connect $\pi^s_i(k)$ to $\pi^t_i(k)$ with a private edge of $G_2$.
Further, for each $1 \leq k \leq a_i$, add a private edge of $G_1$ joining either $\pi^s_i(k)$ to $\pi^t_i(k+1)$ or $\pi^t_i(k)$ to $\pi^s_i(k+1)$ depending on whether $k$ is odd or even, respectively.
We conclude this construction by introducing the concepts of \emph{tunnel} and of \emph{width} of a slice.
The \emph{tunnel} $\Delta_i$ is the subgraph of $S_i$ induced by the vertices of $\pi^t_i$ and $\pi^s_i$, i.e. $\Delta_i = S_i[V(\pi^s_i) \cup V(\pi^t_i)]$.
It is straightforward to see that every tunnel is a biconnected internally-triangulated outer-plane graph, its weak dual is a path, and it contains exactly $2a_i$ triangles.
The \emph{width} $w(S_i)$ of a slice $S_i$ is defined as half the number of triangles in its tunnel.

It is not difficult to see that the transformed instance of \gracsimdrawing{} contains $6Bm + 21m + 7$ vertices and $10Bm + 20m + 7$ edges, therefore its construction can be performed in polynomial time.
We observe that the common subgraph is not connected.
Indeed, $G$ consists of the pumpkin $G_p$ along with all fans and all inner vertices of the transversal paths; thus, there are $2Bm$ isolated vertices in the common subgraph.
Moreover, even $G_1$ and $G_2$ are not connected, because in addition to $G$ they also contain their own private edges of slices $S_i$ ($1 \leq i \leq 3m$) and those of transversal paths $\pi_j$ ($1 \leq j \leq m$); in particular, due to the latter paths, $G_1$ and $G_2$ contain an induced matching of $(B - 1)m$ and $Bm$ (private) edges, respectively.

\noindent
\textsc{Correctness} We now prove that a \emph{Yes}-instance of \threep{} is transformed into a \emph{Yes}-instance of \gracsimdrawing{}, and vice-versa.

($\Rightarrow$) Let $A$ be a \emph{Yes}-instance of \threep{}, we show how to compute a \gracsimdraw{} of the transformed instance $\langle G_1, G_2 \rangle$ on an integer grid; it suffices to compute the vertex coordinates, because edges are represented by straight-line segments.
The drawing construction strongly relies on the concepts of \emph{square cell} and of \emph{cell array}.
A \emph{square cell}, or briefly a \emph{cell}, is a $4 \times 4$ square, with corners at grid points, and with opposite sides that are either horizontal or vertical.
The diagonal of a cell connecting the bottom-left (top-left, respectively) and the top-right (bottom-right, respectively) corners is called the \emph{positive-slope diagonal} (\emph{negative-slope diagonal}).
The \emph{center} of a cell is the intersection point of its diagonals, which meet at right angles.
Every cell contains four special grid points, called \emph{anchor points}, which are the corners of a $2 \times 2$ square having the same center as the cell; two anchor points lie on the positive-slope diagonal while the other two are on the negative-slope diagonal.
A \emph{horizontal cell array} $CA$ of \emph{length} $l > 0$ is an ordered sequence $c_1, c_2, \ldots, c_l$ of $l$ cells such that any two consecutive cells $c_p$, $c_{p+1}$ ($1 \leq p < l$) share a vertical side; namely, the right side of $c_p$ coincide with the left side of $c_{p+1}$.

Consider now a solution $\{A_1, A_2, \ldots, A_m\}$ of \threep{} for the instance $A$.
For each triple $A_j$ ($1 \leq j \leq m$), denote its elements by $a_{j1}$, $a_{j2}$, $a_{j3}$, i.e. $A_j = \{a_{j1}, a_{j2}, a_{j3}\} \subset A$, and denote by $S_{j1}$, $S_{j2}$ and $S_{j3}$, and by $\Delta_{j1}$, $\Delta_{j2}$ and $\Delta_{j3}$, the corresponding slices and their tunnels in the transformed instance.
Embed each tunnel $\Delta_{jk}$ ($1 \leq k \leq 3$) on a horizontal array $CA_{jk}$ of length $a_{jk}$ in such a way that the private edges of $G_2$ are represented by the vertical sides of cells in $CA_{jk}$.
The private edges of $G_1$ are thus embedded on a sequence of $a_{jk}$ cell diagonals, whose slopes are alternately $+1$ (positive-slope diagonal) and $-1$ (negative-slope diagonal), starting from $+1$; hence, in every cell, the anchor points of one of the two diagonals are \emph{occupied}, i.e. they overlap with a straight-line segment representing a private edge of $G_1$, while the remaining two anchor points are (still) \emph{free}.

Place cell arrays $CA_{jk}$ one after another, from left to right, in increasing order of $j = 1, 2, \ldots, m$ and, in case of ties, in increasing order of $k = 1, 2, 3$.
Also, leave a horizontal gap of one cell between intra-partition consecutive arrays and a horizontal gap of two cells in case of inter-partition consecutive arrays.
Concerning the vertical placement proceed as follows.
Let $CA$ and $CA'$ be two arbitrary consecutive arrays (intra- or inter-partition), with $CA$ to the left of $CA'$.
If $CA$ has an even length, then $CA$ and $CA'$ are top- and bottom-aligned along the vertical axis, while if $CA$ has an odd length, then $CA'$ is shifted down of half a cell with respect to $CA$.
It follows that the rightmost free anchor point of $CA$ is always horizontally aligned with the leftmost free anchor point of $CA'$.
Now, let $R$ be the smallest rectangle containing all previous cell arrays with a top, right, bottom, and left margin of one cell.
Place pole $t$ ($s$, respectively) at a grid point above (under, respectively) the top side (bottom side, respectively) of $R$, as close as possible to its vertical bisector line, leaving a vertical offset of two cells; in Fig.~\ref{fi:gracsim-drawing} we deliberately increased this offset to get a better aspect ratio.
Place vertex $v_j$ ($0 \leq j < m$) at the grid point that is horizontally aligned with and to the left of the first free anchor point of $CA_{j1}$, leaving a margin of one cell; also, place vertex $v_m$ at the grid point that is horizontally aligned with and one-cell to the right of the rightmost free anchor point.
Observe that $v_0$ and $v_m$ lie on the left and right side of $R$, respectively.
Now, embed the vertices $v_j^t$ and $v_j^s$ ($j = 0, 1, \ldots, m$) of the pumpkin $G_p$ along the top and bottom side of $R$, respectively, in such a way that they are vertically aligned with $v_j$.
Then, embed the missing vertices of the slices in an analogous way, that is a vertex adjacent to $t$ ($s$, respectively) must be vertically aligned with its neighbor in the tunnel and must lie along the top side (bottom side, respectively) of $R$.
Concerning the handle $\pi_h$, place its subdivision vertex adjacent to $v_0$ at the point whose $x$- and $y$-coordinates are one cell to the left of $v_0$ and one cell above $t$, respectively; with a symmetrical argument choose the position of the other subdivision vertex of $\pi_h$.
It is not hard to see that (\emph{i}) no crossing has been introduced so far; (\emph{ii}) slices $S_{j1}$, $S_{j2}$ and $S_{j3}$ are within wedge $W_j$ ($1 \leq j \leq m$); and (\emph{iii}) every triangle in a tunnel contains exactly one free anchor vertex.
To complete the drawing, it remains to embed the inner vertices of transversal paths, taking into account that every path $\pi_j$ will unavoidably cross the three slices in its edge $W_j$.
Place these vertices at the free anchor points, so that the $p$-th inner vertex of $\pi_j$ occupies the $p$-th free anchor point, from left to right.
It turns out that the produced crossings will always occur at right angles and involve a private edge of $G_1$ and a private edge of $G_2$.
Note that this is possible because, by construction, $w(W_j) = B = w(S_{j1}) + w(S_{j2}) + w(S_{j2})$, where $w(W_j)$ is the \emph{width} of wedge $W_j$, which is defined as the effective length of $\pi_j$.
Indeed, $\pi_j$ has $2B$ inner vertices, there are $2(a_{j1} + a_{j2} + a_{j3})$ free anchor points in $W_j$, and $a_{j1} + a_{j2} + a_{j3} = B$, since we start from a solution of \threep{}.

($\Leftarrow$) Let $\langle \Gamma_1, \Gamma_2 \rangle$ be any \gracsimdraw{} of $\langle G_1, G_2 \rangle$, and let $\Gamma_p$ be the drawing of $G_p$ induced by $\langle \Gamma_1, \Gamma_2 \rangle$.
Also, let $C_j \subset G_p$ ($1 \leq j \leq m$) be the cycle consisting of paths $\pi_s(j-1)$, $\pi_t(j-1)$, $\pi_t(j)$ and $\pi_s(j)$.
We first claim that the following invariants are satisfied.
(\emph{I1}) $C_j$ ($1 \leq j \leq m$) is the boundary of a wedge $W_j$ in $\Gamma_p$, where a wedge is a bounded or unbounded face of degree eight in $\Gamma_p$.
(\emph{I2}) Transversal path $\pi_j$ ($1 \leq j \leq m$) is drawn within wedge $W_j$.
(\emph{I3}) Any two slices cannot be contained one in another and do not overlap with each other except at poles $s$ and $t$.
(\emph{I4}) Every edge of $\pi_j$ ($1 \leq j \leq m$) crosses at most one edge of a same slice.
(\emph{I5}) Every wedge contains exactly three slices.

Let $R_b(C_j)$ and $R_u(C_j)$ be the bounded and the unbounded plane regions, respectively, delimited by $C_j$ in $\Gamma_p$.
Since $v_{j-1}$ and $v_j$ are two vertices of $C_j$, path $\pi_j$ has to be drawn within either $R_b(C_j)$ or $R_u(C_j)$, otherwise an inner edge of $\pi_j$ would cross an edge of $C_j$, which is not allowed in a \gracsimdraw{} of $\langle G_1, G_2 \rangle$ because $C_j \subset G$.
Also, if $\pi_j$ is contained in $R_b(C_j)$ ($R_u(C_j)$, respectively), then all the other paths of the pumpkin that connect the two poles $s$ and $t$ must be drawn within $R_u(C_j)$ ($R_b(C_j)$, respectively).
Invariants \emph{I1} and \emph{I2} are thus satisfied.
Concerning invariant \emph{I3}, it is immediate to see that any two slices cannot be contained one in another.
Further, in case of overlap, an edge $e_1$ of a slice $S_1$ would cross a boundary edge $e_2$ of a slice $S_2$, where $e_2$ is a private edge of $G_2$ and $e_1$ is a private edge of $G_1$.
But this is not possible, because the end-vertices of $e_1$ are also connected in $S_1$ by a $2$-edge path consisting of a shared edge and of a private edge of $G_2$.
Invariant \emph{I4} holds because every transversal path $\pi_j$ ($1 \leq j \leq m$) can only cross edges of tunnels in $W_j$, and every tunnel is drawn as a straight-line internally triangulated outer-plane graph.
Therefore, $\pi_j$ cannot enter and then exit from a triangle with a same private edge in such a way that all edge crossings are at right angles.
Namely, every triangle of a tunnel in $W_j$ takes at least one inner vertex of $\pi_j$.
We now show that invariant \emph{I5} is satisfied.
It is straightforward to see that every slice must be drawn within some wedge, and all the slices in a wedge $W_j$ are crossed by its transversal path $\pi_j$.
In particular, $\pi_j$ has to pass through the tunnels of these slices and such tunnels are pairwise disjoint and none of them contains another.
Suppose by contradiction that invariant \emph{I5} does not hold.
Then, there would be a wedge $W_p$ ($1 \leq p \leq m$) containing at least four slices; recall that there are $3m$ slices to be distributed among $m$ wedges.
Let us denote such slices by $S_{p1}, S_{p2}, \ldots, S_{pk}$, with $k \geq 4$, and let $a_{pl} \in A$ be the integer encoded by slice $S_{pl}$ ($1 \leq l \leq k$).
Since each element of $A$ is strictly greater than $B/4$, it follows that $\sum_{l = 1}^{k} w(S_{pl}) = \sum_{l = 1}^{k} a_{pl} > \sum_{l = 1}^{k} B/4 \geq B = w(W_p)$, thus wedge $W_p$ is not wide enough to host all its slices, a contradiction.
In other words, the alternating path $\pi_p$ does not have enough inner vertices to pass through all the tunnels of slices in $W_p$ avoiding crossing that are not allowed in a \gracsimdraw{}.\\
Now, for each wedge $W_j$ ($1 \leq j \leq m$), denote by $S_{j1}$, $S_{j2}$ and $S_{j3}$ the three slices that are within $W_j$, and let $a_{j1}$, $a_{j2}$ and $a_{j3}$ be their corresponding elements of $A$. 
We claim that $a_{j1} + a_{j2} + a_{j3} = B$.
Indeed, it cannot be $\sum_{k = 1}^{3} a_{jk} > B$, because it would imply that $\sum_{k = 1}^{3} w(S_{jk}) > w(W_j)$, which is not possible as seen above.
On the other hand, if $\sum_{k = 1}^{3} a_{jk} < B$, there would be some $j' \neq j$ with $1 \leq j' \leq m$ such that $\sum_{k = 1}^{3} a_{j'k} > B$, otherwise $\sum_{i = 1}^{3m} a_i$ would be strictly less than $mB$, which violates our initial hypothesis on the elements of $A$.
Hence, even this case is not possible.
In conclusion, every wedge $W_j$ ($1 \leq j \leq m$) contains exactly three slices $S_{j1}$, $S_{j2}$ and $S_{j3}$, each of these slices has a width $w(S_{jk})$ ($1 \leq k \leq 3$) that encodes a distinct element of $A$, and the sum of these widths is equal to $B$, i.e. $w(S_{j1}) + w(S_{j2}) + w(S_{j3}) = B$.
Therefore, the partitioning of $A$ defined by $A_1, A_2, \ldots, A_m$, where $A_j = \{w(S_{j1}), w(S_{j2}), w(S_{j3})\}$, is a solution of \threep{} for the instance $A$. \hfill$\qed$
\end{proof}

\noindent
We conclude this section with two remarks.

\begin{remark}
It is not hard to see that this reduction can also be used to give an alternative proof for the \nphness{} of \textsc{SGE}, which was proved by Estrella-Balderrama et al.~\cite{egjpss-sgge-07}.
\end{remark}

\begin{remark}
It is not clear whether this reduction can be adapted to study the complexity of the \emph{one bend extension} of \textsc{GRacSim}, i.e. the variant of \textsc{GRacSim} in which one bend per edge is allowed; we leave this question as open problem.
\end{remark}

\section{\npcness{} of \ksefe{}}
\label{sec:ksefenpc}

In order to increase the readability of a simultaneous embedding, which is particularly desired in graph drawing applications, one may wonder whether it is possible to compute a \sefe{}, where every private edge receives at most a limited and fixed number of crossings.
We recall that there is no restriction on the number of crossings that involve a private edge in a \sefe{} drawing.
Further, two private edges may cross more than once, and these multiple crossings could be necessary for the existence of a simultaneous embedding; however, Frati et al.~\cite{fhk-sefbc-15} have shown that whenever two planar graphs admit a \sefe{}, then they also admit a \sefe{} with at most sixteen  crossings per edge pair.

Motivated by the previous considerations, we introduce and study the complexity of the following problem, named \ksefe{}, where $k$ denotes a fixed bound on the number of crossings per edge that are allowed.\\
\\
\noindent
\begin{minipage}[t]{0.15\textwidth}
\raggedleft
\textsf{Problem}:
\end{minipage}
\begin{minipage}[t]{0.5\textwidth}
\end{minipage}
\begin{minipage}[t]{0.80\textwidth}
\raggedright
\ksefe{}
\end{minipage}

\noindent
\begin{minipage}[t]{0.15\textwidth}
\raggedleft
\textsf{Instance}:
\end{minipage}
\begin{minipage}[t]{0.5\textwidth}
\end{minipage}
\begin{minipage}[t]{0.80\textwidth}
\raggedright
Two planar graphs $G_1 = (V, E_1)$ and $G_2 = (V, E_2)$, sharing a common subgraph $G = (V, E) = (V, E_1 \cap E_2)$, and a positive integer $k$.
\end{minipage}

\noindent
\begin{minipage}[t]{0.15\textwidth}
\raggedleft
\textsf{Question}:
\end{minipage}
\begin{minipage}[t]{0.5\textwidth}
\end{minipage}
\begin{minipage}[t]{0.80\textwidth}
\raggedright
Do $G_1$ and $G_2$ admit a \sefe{} such that every private edge receives at most $k$ crossings?
\end{minipage}
\\
\\
\noindent
It is straightforward to see that \ksefe{} is, in general, a restricted version of \sefe{}.
Namely, for any positive integer $k$, it is easy to find pairs of graphs that admit a \textsc{(k+1)-SEFE}, and thus a \sefe{}, but not a \ksefe{}.
For example, consider a pair of graphs $G_1 = (V, E_1)$ and $G_2 = (V, E_2)$ defined as follows (an illustration for $k = 4$ is given in Fig.~\ref{fi:kseferestriction}).
The common subgraph $G = (V, E)$ is a wheel of $2k + 5$ vertices, where $u_0$, $u_1, \ldots u_{k+1}, v_0$, $v_1, \ldots, v_{k+1}$ are the  $2(k + 2)$ vertices of its cycle in clockwise order, $E_1 = E \cup \{(u_0, v_0)\}$, and $E_2 = E \cup \bigcup_{i = 1}^{k + 1}\{(u_i, v_{k+2-i})\}$.
Since $G$ has a unique planar embedding (up to a homomorphism of the plane), the private edge $(u_0, v_0)$ of $G_1$ crosses all the $k+1$ private edges of $G_2$, i.e. all the edges $(u_i, v_{k+2-i})$ with $1 \leq i \leq k+1$.
Therefore, $G_1$ and $G_2$ admit a \textsc{(k+1)-SEFE}, and thus a \sefe{}, but not a \ksefe{}.

\begin{figure}[t]
\centering
\includegraphics[width=0.60\textwidth]{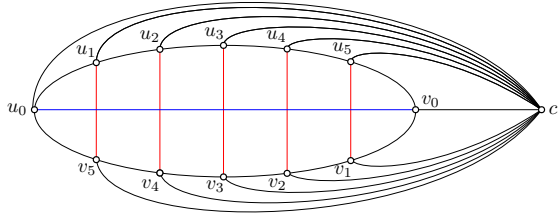}
\caption{
A pair of graphs that admit a \ksefe{} only for $k \geq 5$.
}
\label{fi:kseferestriction}
\end{figure}

\begin{theorem}\label{th:1sefenph}
\onesefe{} is~\nph{}.
\end{theorem}

\begin{proof}
We use a reduction from \threep{} similar to that in the proof of Theorem~\ref{th:gracsimnph}; subdivision vertices are now omitted, since we are no longer in a geometric setting.\\

\noindent
\textsc{Construction} Start with a (non-subdivided) pumpkin $G_p \subset G$ whose vertices $v_0$, $v_1, \ldots v_m$ are adjacent to the two poles $s$ and $t$, and whose handle is a single edge $(v_0, v_m)$.
Add a transversal path $\pi_j$ between every pair of vertices $v_{j-1}$ and $v_j$ ($1 \leq j \leq m$).
Differently from the proof of Theorem~\ref{th:gracsimnph}, $\pi_j$ has to contain $2B - 1$ inner vertices and not $2B$; the reason of this will be clarified later.
Also, the effective length of $\pi_j$ is now defined as half the number of its edges, hence it is still equal to $B$.
Slice gadgets $S_i$ ($1 \leq i \leq 3m$) and their tunnels $\Delta_i$ are also slightly modified and are defined as follows.
For each integer $a_i \in A$, create an alternating path $\pi(S_i)$ of $2a_i$ non-shared edges; thus, $\pi(S_i)$ has $2a_i + 1$ vertices and its extremal edges never belong to the same graph $G_i$ ($i = 1,2$).
Construct a fan $F_i^t$ by adding an edge between all the pairs of consecutive vertices of $\pi(S_i)$ in even positions and by connecting such vertices to the pole $t$ of the pumpkin; $F_i^t \setminus \{t\}$ is a path of $a_i - 1$ edges, because $\pi(S_i)$ has $a_i$ vertices in even positions and $a_i + 1$ vertices in odd positions.
Similarly, construct a fan $F_i^s$ by connecting the pole $s$ with a path of $a_i$ edges passing through all the vertices of $\pi(S_i)$ in odd positions.
Slice $S_i$ is composed from the two fans $F_i^t$ and $F_i^s$ plus all the edges of $\pi(S_i)$.
Further, all the edges of fans are shared, while those of $\pi(S_i)$ are not shared and belong alternatively to $G_2$ and to $G_1$.
The tunnel $\Delta_i$ of a slice $S_i$ is the subgraph that results from $S_i$ after removing the two poles $s$ and $t$, i.e $\Delta_i = S_i \setminus \{s,t\}$.
It is straightforward to see that every tunnel is a biconnected internally-triangulated outer-plane graph, whose weak dual is a path, and it contains exactly $2a_i - 1$ triangles if the corresponding slice encodes integer $a_i$.
The width $w(S_i)$ of a slice $S_i$ is defined as half the number of private edges in its tunnel $\Delta_i$, thus $w(S_i) = a_i$.

\begin{figure}[!tb]
\centering
\subfigure[Pumpkin gadget]{
  \label{fi:1-sefe-pumpkin}  
  \includegraphics[width=0.56\textwidth]{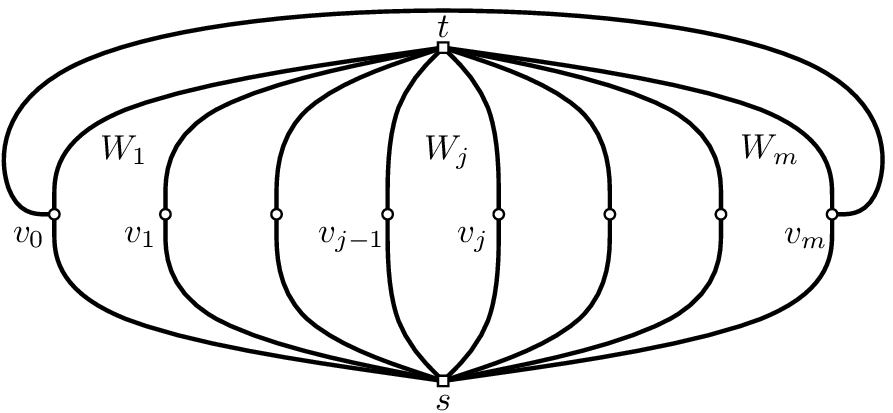}
}
\hspace{2mm}
\subfigure[Slice gadget]{
  \label{fi:1-sefe-slice}
  \includegraphics[width=0.30\textwidth]{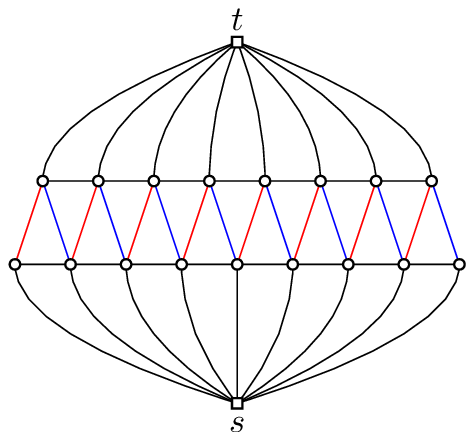}
}
\subfigure[Wedge $W_j$, transversal path $\pi_j$, and slices $S_{j1}$, $S_{j2}$, and $S_{j3}$]{
  \label{fi:1-sefe-wedgezoom}
  \includegraphics[width=0.99\textwidth]{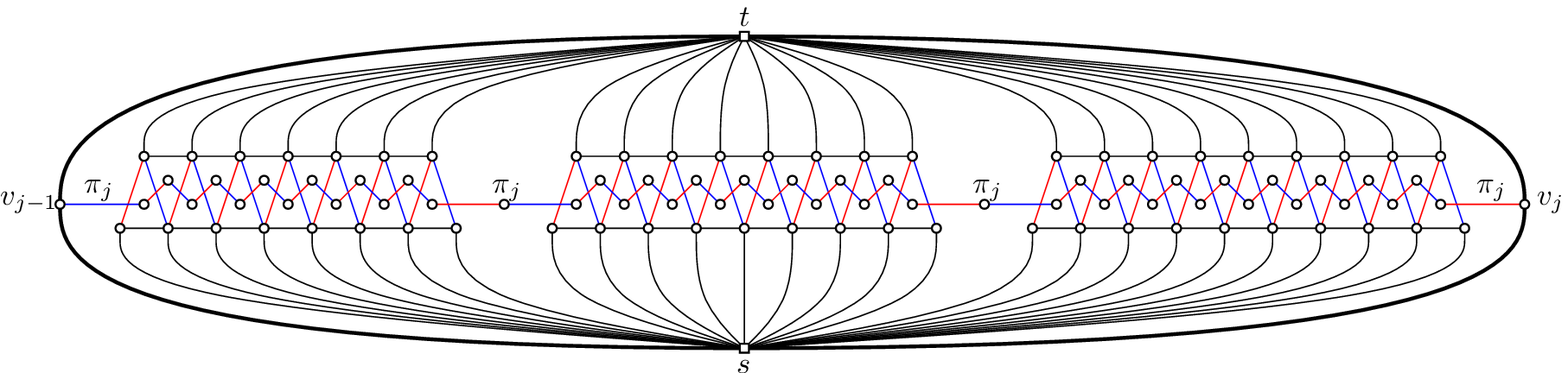}
}
\subfigure[\textsc{1-SEFE} drawing]{
  \label{fi:1-sefe-drawing}
  \includegraphics[width=0.99\textwidth]{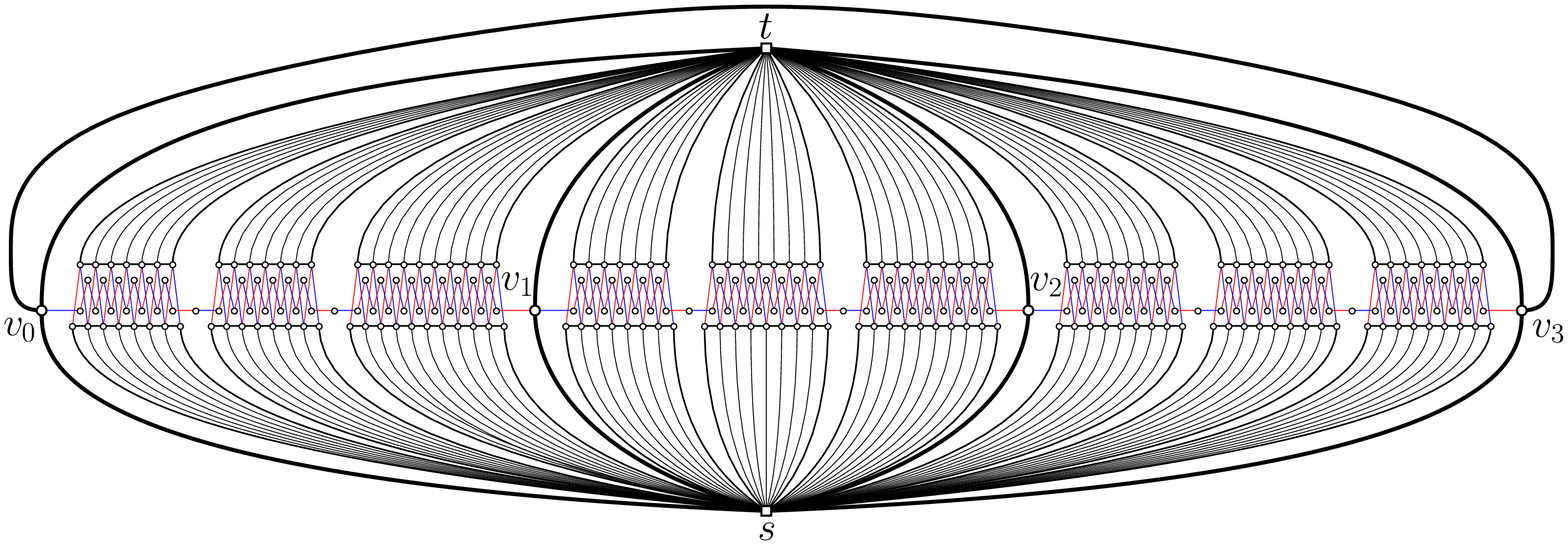}
}
\caption{
Illustration of (a) a pumpkin gadget and of (b) a slice gadget encoding integer $8$.
(c) A wedge $W_j$ of width~$24$, its transversal path $\pi_j$, and slices $S_{j1}$, $S_{j2}$, and $S_{j3}$ encoding integers~$7$,~$8$ and $9$, respectively.
Shared edges are colored black, those of the pumpkin with thick lines, while private edges of $G_1$ and of $G_2$ are colored blue and red, respectively.
(d) A \onesefe{} drawing of the transformed instance of \threep{}, when $m = 3$, $B = 24$ and $A = \{7,7,7,8,8,8,8,9,10\}$. Slices are drawn within wedges according to the following solution of \threep{}: $A_1 = \{7,7,10\}$, $A_2 = \{7,8,9\}$ and $A_3 = \{8,8,8\}$.
}
\label{fi:1-sefe-pumpkin-tpath-slice-wedgezoom}
\end{figure}

It is not hard to see that the transformed instance $\langle G_1, G_2 \rangle$ contains $4Bm + 9m + 3$ vertices and $8Bm + 2m + 3$ edges, thus its construction can be done in polynomial time.
Furthermore, we observe that $G$, $G_1$ and $G_2$ are not connected.
Indeed, $G$ contains $(2B - 1)m$ isolated vertices, i.e. all the inner vertices of transversal paths, while $G_1$ and $G_2$ contain an induced matching of $(B - 1)m$ (private) edges each.\\

\noindent
\textsc{Correctness} Let $A$ be an instance of \threep{}, and let $\langle G_1, G_2 \rangle$ be an instance of \onesefe{} obtained by using the previous transformation.
We show that $A$ admits a $3$-partition if and only if $\langle G_1, G_2 \rangle$ admits a \onesefe{} drawing.

($\Rightarrow$) Suppose that $A$ admits a $3$-partition $\{A_1, A_2, \ldots, A_m\}$, then a \onesefe{} drawing of $\langle G_1, G_2 \rangle$ can be constructed as follows.
Compute a plane drawing $\Gamma_p$ of the pumpkin $G_p$ (see, e.g., Fig.~\ref{fi:1-sefe-pumpkin}) such that (\emph{i}) the external face is delimited by the edges $(s, v_0)$, $(v_0, v_m)$ and $(v_m, s)$ and (\emph{ii}) for each $j = 1, 2, \ldots, m$ edge $(t, v_j)$ immediately follows edge $(t, v_{j-1})$ in the counterclockwise edge ordering around $t$.
$\Gamma_p$ contains $m$ inner faces of degree four, delimited by edges $(s, v_{j-1})$, $(v_{j-1}, t)$, $(t, v_j)$, $(v_j, s)$ ($1 \leq j \leq m$), which are the wedges $W_j$ of the pumpkin.
Consider now each triple $A_j = \{a_{j1}, a_{j_2}, a_{j3}\}$ ($1 \leq j \leq m$), and denote by $S_{j1}$, $S_{j2}$, $S_{j3}$ the corresponding slices in the transformed instance.
For each slice $S_{jk}$ ($1 \leq k \leq 3$), compute a plane drawing with both poles on the external face.
Place these drawings one next to the other within wedge $W_j$, in any order; for simplicity we may assume that $S_{j1}$ is the leftmost slice, $S_{j2}$ is the middle slice and $S_{j3}$ is the rightmost one.
Also, if necessary, flip each slice around its poles so that the leftmost private edge always belongs to $G_2$; clearly, this implies that the rightmost private edge belongs to $G_1$.
It is not difficult to see that the drawing produced so far is planar, i.e. even the private edges do not create crossings.
Moreover, since $w(W_j) = B = a_{j1} + a_{j2} + a_{j3} = w(S_{j1}) + w(S_{j2}) + w(S_{j3})$, every transversal path $\pi_j$ ($1 \leq j \leq m$) can be drawn within wedge $W_j$ in such a way that (\emph{i}) every edge of $\pi_j$ crosses exactly one private edge of a tunnel in $W_j$, and (\emph{ii}) every crossing involves a private edge of $G_1$ and a private edge of $G_2$.\\

($\Leftarrow$) We conclude the proof by showing that if $\langle G_1, G_2 \rangle$ admits a \onesefe{} drawing $\langle \Gamma_1, \Gamma_2 \rangle$, then $A$ admits a $3$-partition.
By a similar argument as that in the proof of Theorem~\ref{th:gracsimnph}, $\langle \Gamma_1, \Gamma_2 \rangle$ induces a plane drawing $\Gamma_p$ of the pumpkin $G_p$, in which each wedge $W_j$, i.e. each bounded or unbounded face of degree four of $G_p$, is delimited by a cycle $C_j$ consisting of edges $(s, v_{j-1})$, $(v_{j-1}, t)$, $(t, v_j)$ and $(v_j, s)$, for some $1 \leq j \leq m$.
Further, path $\pi_j$ has to be drawn within $W_j$, and for each $1 \leq i \leq 3m$, fans $F_i^t$ and $F_i^s$, and thus the slice $S_i$ they belong to must be placed within a same wedge.
Let $S_{j1}, S_{j2}, \ldots, S_{jk}$ be the slices within wedge $W_j$, for some $k \geq 0$.
Since every private edge receives at most one ($k = 1$) crossing in $\langle \Gamma_1, \Gamma_2 \rangle$, it follows that $\sum_{l = 1}^{k} w(S_{jl}) \leq w(W_j) = B$, i.e. the number of edges of $\pi_j$ must be greater than or equal to the number of edges of tunnels in $W_j$.
We now show that there are exactly three slices in every wedge, i.e. $k = 3$.
It cannot be $k > 3$, otherwise $\sum_{l = 1}^{k} w(S_{jl}) = \sum_{l = 1}^{k} a_{jl} > \sum_{l = 1}^{k} B/4 \geq B = w(W_j)$.
On the other hand, it cannot be $k < 3$, otherwise there would some other wedge with $k' > 3$ slices; recall that there are a total of $3m$ slices and a total of $m$ wedges.
Suppose now that $\sum_{l = 1}^{3} w(S_{jl}) < w(W_j) = B$, for some $1 \leq j \leq m$.
Then, there would exist some $j' \neq j$ with $1 \leq j' \leq m$ such that $\sum_{l = 1}^{3} w(S_{j'l}) > w(W_{j'}) = B$, otherwise it would be violated the equality $\sum_{i = 1}^{3m} a_i = mB$.
In conclusion, there are exactly three slices in every wedge, and the sum of their widths coincides with $B$.
Therefore the partitioning $A_1, A_2, \ldots, A_m$ of $A$, where $A_j = \{w(S_{j1}), w(S_{j2}), w(S_{j3})\}$, is a $3$-partition.\hfill$\qed$
\end{proof}

\begin{theorem}\label{th:ksefenpc}
For any fixed $k \geq 1$, \ksefe{} is~\npc{}.
\end{theorem}

\begin{proof}
Concerning the \nphness{}, it suffices to repeat the proof of Theorem~\ref{th:1sefenph}, by replacing every private edge $e$ of each tunnel of $G_i$ ($i=1,2$) with a set of $k$ internally vertex-disjoint paths $\pi_1(e)$, $\pi_2(e), \ldots, \pi_k(e)$, consisting each one of two private edges of $G_i$.

We now introduce some definitions and then prove the membership in \np{} using an approach similar to that described in~\cite{gj-cnnpc-83}.
An \emph{edge crossing structure} $\chi(e_1)$ of a private edge $e_1 \in E_1$ is a pair $\langle \varepsilon_2, \sigma(\varepsilon_2) \rangle$, where $\varepsilon_2$ is a multiset on the set $E_2 \setminus E$ with cardinality at most $k$, and $\sigma(\varepsilon_2)$ is a permutation of multiset $\varepsilon_2$.
A \emph{crossing structure} $\chi(G_1,G_2)$ of a pair of graphs $\langle G_1, G_2 \rangle$ is an assignment of an edge crossing structure to each private edge of $E_1$.
Of course, all crossing structures of $\langle G_1, G_2 \rangle$ can be non-deterministically generated in a time that is polynomial in $|V| = n$, and they include the crossing structures induced by all \ksefe{} drawings of $\langle G_1, G_2 \rangle$.
We conclude the proof by describing a polynomial time algorithm for testing whether a given crossing structure $\chi(G_1,G_2)$ is a crossing structure induced by some \ksefe{} drawing of $\langle G_1, G_2 \rangle$.
Let $G_{\cup}$ be the union graph of $G_1$ and $G_2$, i.e. $G_{\cup} = (V, E_1 \cup E_2)$.
For each edge $e$ of $G_{\cup}$ such that $e \in E_1 \setminus E$, consider its crossing structure $\chi(e) = \langle \varepsilon_2, \sigma(\varepsilon_2) \rangle$, replace every crossing between $e$ and the edges in $\varepsilon_2$ with a dummy vertex, preserving the ordering given by $\sigma(\varepsilon_2)$, and then test the resulting (multi) graph for planarity.\hfill$\qed$
\end{proof}

\noindent
We conclude even this section with two remarks.

\begin{remark}
The previous reduction cannot be successfully applied to \sefe{}, because of the \emph{$2$-edge penetration vulnerability}: every transversal path $\pi_j$ ($1 \leq j \leq m$) can pass through all the tunnels in $W_j$ using only its two first edges; an illustration of this vulnerability is given in Fig.~\ref{fi:2epv-ksefe}.
Also, any tentative to patch this vulnerability by replacing the transversal paths with different graphs, modifying the slices accordingly, always resulted in constructions in which overlapping slices were possible.
\end{remark}

\begin{remark}
From a theoretical point of view, it also makes sense to study a slightly different restriction of \sefe{}, where instead of limiting the number of crossings per edge, it is limited the number of distinct edges that cross a same private edge; recall that two private edges may cross each other more than once, which gives rise to a different problem than \ksefe{}.
We may call this problem \kpairsefe{}, because $k$ is now the bound on the allowed number of crossing edge pairs involving a same edge.
It is not hard to see that a reduction analogous to that given in the proof of Theorems~\ref{th:1sefenph}~and~\ref{th:ksefenpc} can be used to prove the \nphness{} of \kpairsefe{}.
The interesting theoretical aspect of \kpairsefe{} is the following: if $k$ is greater than or equal to the maximum number of edges of $G_i$ ($i=1,2$), then a \kpairsefe{} is also a \sefe{}; in particular, if $k \geq 3|V| - 6$ the two problems are identical.
\end{remark}

\begin{figure}[t]
\centering
\includegraphics[width=0.80\textwidth]{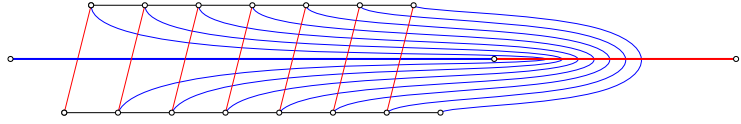}
\caption{
Illustration of the $2$-edge penetration vulnerability.
}
\label{fi:2epv-ksefe}
\end{figure}


\newpage
\section{Conclusions and Open Problems}
In this work we have shown the \nphness{} of the \gracsimdrawing{} problem, a restricted version of the \sge{} problem in which edge crossings must occur only at right angles.
Then, we have introduced and studied the \npcness{} of the \ksefe{} problem, a restricted version of the \sefe{} problem, where every private can receive at most $k$ crossings.

Our results raise two main questions.
First, as already mentioned at the end of Section~\ref{sec:gracsimnph}, it would be interesting to study the complexity of a relaxed version of the \gracsimdrawing{} problem, where a prescribed number of bends per edge are allowed; this open problem was already posed in~\cite{bdkw-sdpgrac-16}. In particular, it is not clear whether the reduction given in the proof of Theorem~\ref{th:gracsimnph} can be adapted for proving the \nphness{} of the one bend extension of \gracsim{}.
Another interesting open problem is to investigate the complexity of \kpairsefe{} when the ratio $|V|/k$ tends to $\frac{1}{3} + \frac{2}{k}$ from the right; we recall that for $k \geq 3|V| - 6$, \kpairsefe{} and \sefe{} are the same problem, and that the \nphness{} of \kpairsefe{} strongly relies on a construction where the ratio $|V|/k$ is significantly greater than $\frac{1}{3} + \frac{2}{k}$.

\bibliographystyle{splncs03}
\bibliography{Hardness-GRacSim-kSefe}

\begin{thebibliography}{10}
\providecommand{\url}[1]{\texttt{#1}}
\providecommand{\urlprefix}{URL }

\bibitem{ad-scp-14}
Angelini, P., {Da Lozzo}, G.: Sefe = c-planarity? In: ICGT 2014 (2014)

\bibitem{adf-seepg-13}
Angelini, P., {Di Battista}, G., Frati, F.: Simultaneous embedding of embedded
  planar graphs. Int. J. Comput. Geometry Appl.  23(2),  93--126 (2013)

\bibitem{adfjkpr-tppeg-10}
Angelini, P., {Di Battista}, G., Frati, F., Jel{\'{\i}}nek, V.,
  Kratochv{\'{\i}}l, J., Patrignani, M., Rutter, I.: Testing planarity of
  partially embedded graphs. In: Charikar, M. (ed.) SODA 2010. pp. 202--221.
  ACM-SIAM, SIAM (2010)

\bibitem{adfpr-tsetgibcg-12}
Angelini, P., {Di Battista}, G., Frati, F., Patrignani, M., Rutter, I.: Testing
  the simultaneous embeddability of two graphs whose intersection is a
  biconnected or a connected graph. J. Discrete Algorithms  14,  150--172
  (2012)

\bibitem{aefg-smliogpg-16}
Angelini, P., Evans, W.S., Frati, F., Gudmundsson, J.: {SEFE} without mapping
  via large induced outerplane graphs in plane graphs. Journal of Graph Theory
  82(1),  45--64 (2016)

\bibitem{agkn-tpgse-12}
Angelini, P., Geyer, M., Kaufmann, M., Neuwirth, D.: On a tree and a path with
  no geometric simultaneous embedding. J. Graph Algorithms Appl.  16(1),
  37--83 (2012)

\bibitem{adn-asefepbep-15}
Angelini, P., Lozzo, G.D., Neuwirth, D.: Advancements on {SEFE} and partitioned
  book embedding problems. Theor. Comput. Sci.  575,  71--89 (2015)

\bibitem{abks-gracsdg-13}
Argyriou, E.N., Bekos, M.A., Kaufmann, M., Symvonis, A.: Geometric {RAC}
  simultaneous drawings of graphs. J. Graph Algorithms Appl.  17(1),  11--34
  (2013)

\bibitem{bdkw-sdpgrac-16}
Bekos, M.A., van Dijk, T.C., Kindermann, P., Wolff, A.: Simultaneous drawing of
  planar graphs with right-angle crossings and few bends. J. Graph Algorithms
  Appl.  20(1),  133--158 (2016)

\bibitem{bkr-seeorpc-13}
Bl{\"{a}}sius, T., Karrer, A., Rutter, I.: Simultaneous embedding: Edge
  orderings, relative positions, cutvertices. In: Wismath, S., Wolff, A. (eds.)
  GD 2013. LNCS, vol. 8242, pp. 220--231. Springer, Heidelberg (2013)

\bibitem{bkr-hgdv-13}
Bl{\"{a}}sius, T., Kobourov, S.G., Rutter, I.: Simultaneous embedding of planar
  graphs. In: Tamassia, R. (ed.) Handbook of Graph Drawing and Visualization,
  chap.~11, pp. 349--381. CRC (2013)

\bibitem{br-spqoacep-13}
Bl{\"{a}}sius, T., Rutter, I.: Simultaneous pq-ordering with applications to
  constrained embedding problems. In: Khanna, S. (ed.) SODA 2010. pp.
  1030--1043. ACM-SIAM, SIAM (2013)

\bibitem{cfglms-dpespg-15}
Chan, T.M., Frati, F., Gutwenger, C., Lubiw, A., Mutzel, P., Schaefer, M.:
  Drawing partially embedded and simultaneously planar graphs. J. Graph
  Algorithms Appl.  19(2),  681--706 (2015)

\bibitem{ek-sepgfb-05}
Erten, C., Kobourov, S.G.: Simultaneous embedding of planar graphs with few
  bends. J. Graph Algorithms Appl.  9(3),  347--364 (2005)

\bibitem{ekln-sgdlavs-05}
Erten, C., Kobourov, S.G., Le, V., Navabi, A.: Simultaneous graph drawing:
  Layout algorithms and visualization schemes. J. Graph Algorithms Appl.  9(1),
   165--182 (2005)

\bibitem{egjpss-sgge-07}
Estrella{-}Balderrama, A., Gassner, E., J{\"{u}}nger, M., Percan, M., Schaefer,
  M., Schulz, M.: Simultaneous geometric graph embeddings. In: Hong, S.H.,
  Nishizeki, T., Quan, W. (eds.) GD 2007. LNCS, vol. 4875, pp. 280--290.
  Springer, Heidelberg (2008)

\bibitem{elm-svrpstguls-16}
Evans, W.S., Liotta, G., Montecchiani, F.: Simultaneous visibility
  representations of plane st-graphs using l-shapes. Theor. Comput. Sci.  645,
  100--111 (2016)

\bibitem{f-egsfe-06}
Frati, F.: Embedding graphs simultaneously with fixed edges. In: Kaufmann, M.,
  Wagner, D. (eds.) GD 2006. LNCS, vol. 4372, pp. 108--113. Springer,
  Heidelberg (2006)

\bibitem{fhk-sefbc-15}
Frati, F., Hoffmann, M., Kusters, V.: Simultaneous embeddings with few bends
  and crossings. In: {Di Giacomo}, E., Lubiw, A. (eds.) GD 2015. LNCS, vol.
  9411, pp. 166--179 (2015)

\bibitem{gj-cigtnpc-79}
Garey, M.R., Johnson, D.S.: Computers and Intractability: A Guide to the Theory
  of NP-Completeness. W. H. Freeman \& Co., New York, NY, USA (1979)

\bibitem{gj-cnnpc-83}
Garey, M.R., Johnson, D.S.: Crossing number is np-complete. SIAM Journal on
  Algebraic Discrete Methods  4(3),  312--316 (1993)

\bibitem{gjpss-sgefe-06}
Gassner, E., J{\"{u}}nger, M., Percan, M., Schaefer, M., Schulz, M.:
  Simultaneous graph embeddings with fixed edges. In: Fomin, F.V. (ed.) WG
  2006. LNCS, vol. 4271, pp. 325--335. Springer, Heidelberg (2006)

\bibitem{ghkr-dsegfb-14}
Grilli, L., Hong, S.H., Kratochv{\'{\i}}l, J., Rutter, I.: Drawing
  simultaneously embedded graphs with few bends. In: Duncan, C., Symvonis, A.
  (eds.) GD 2014. LNCS, vol. 8871, pp. 40--51. Springer, Heidelberg (2014)

\bibitem{hjl-tspcg2c-13}
Haeupler, B., Jampani, K.R., Lubiw, A.: Testing simultaneous planarity when the
  common graph is 2-connected. J. Graph Algorithms Appl.  17(3),  147--171
  (2013)

\bibitem{js-igsefe-09}
J{\"{u}}nger, M., Schulz, M.: Intersection graphs in simultaneous embedding
  with fixed edges. J. Graph Algorithms Appl.  13(2),  205--218 (2009)

\bibitem{s-ttphtpv-13}
Schaefer, M.: Toward a theory of planarity: Hanani-tutte and planarity
  variants. J. Graph Algorithms Appl.  17(4),  367--440 (2013)

\end{thebibliography}

\end{document}